\documentclass[11pt]{article}
\usepackage{authblk}

\usepackage{hyperref}
\usepackage{graphicx}
\usepackage[ruled,linesnumbered,algo2e,vlined]{algorithm2e}
\usepackage{amsfonts}
\usepackage{amsmath}
\usepackage{cleveref}

\newcommand{\etal}{et al.}

\newcommand{\AOccur}{$\approx$-occurrence}
\newcommand{\AOccurs}{$\approx$-occurrences}
\newcommand{\ADic}{$\approx$-dictionary matching}
\newcommand{\AEn}{$\approx$-prefix encoding}
\newcommand{\AMatch}{$\approx$-match}
\newcommand{\OCC}{\mathit{occ}}
\newcommand{\Queue}{\mathit{queue}}
\newcommand{\En}[1]{\langle #1 \rangle}
\newcommand{\Null}{\mathsf{NULL}}
\newcommand{\Pref}{\mathsf{Pref}}
\newcommand{\SCERA}{\mathsf{SCERA}}
\newcommand{\AC}{\mathsf{AC}}
\newcommand{\Goto}{\delta}
\newcommand{\Fail}{\mathsf{fail}}
\newcommand{\Out}{\mathsf{out}}
\newcommand{\Depth}{\mathsf{dep}}
\newcommand{\Label}{\mathsf{label}}
\newcommand{\ConstGoto}{\mathsf{constructGoto}}
\newcommand{\ConstFail}{\mathsf{constructFailure}}
\newcommand{\ConstOut}{\mathsf{constructOutput}}
\newcommand{\Match}{\mathsf{Matching}}
\newcommand{\Root}{\mathit{root}}
\newcommand{\DD}{\mathcal{D}}
\newcommand{\CS}{\mathcal{S}}
\newcommand{\EnTime}[1]{\xi(#1)}
\newcommand{\REnTime}[1]{\phi(#1)}

\SetKwProg{Fn}{Function}{}{end}

\usepackage[top=3cm, bottom=3cm, left=3.5cm, right=3.5cm, includefoot]{geometry}
\usepackage{amsthm}
\newtheorem{definition}{Definition}
\newtheorem{lemma}{Lemma}
\newtheorem{theorem}{Theorem}

\begin{document}
\title{Generalized Dictionary Matching under Substring Consistent Equivalence Relations}
%
%\titlerunning{Generalized Dictionary Matching under SCERs}
\author{Diptarama Hendrian}
%
%\authorrunning{Diptarama Hendrian}
\affil{Graduate school of information sciences, Tohoku University, Japan\\
\texttt{diptarama@tohoku.ac.jp}}
\date{}
\maketitle              % typeset the header of the contribution
\begin{abstract}
Given a set of patterns called a dictionary and a text,
the dictionary matching problem is a task to find all occurrence positions of all patterns
in the text.
The dictionary matching problem can be solved efficiently by using the Aho-Corasick algorithm.
Recently, Matsuoka et al. [TCS, 2016] proposed a generalization of pattern matching problem under substring consistent equivalence relations
and presented a generalization of the Knuth-Morris-Pratt algorithm to solve this problem.
An equivalence relation $\approx$ is a substring consistent equivalence relation (SCER)
if for two strings $X,Y$, $X \approx Y$ implies $|X| = |Y|$ and $X[i:j] \approx Y[i:j]$ for all $1 \le i \le j \le |X|$.
In this paper, we propose a generalization of the dictionary matching problem
and present a generalization of the Aho-Corasick algorithm for the dictionary matching under SCER.
We present an algorithm that constructs SCER automata and an algorithm that performs dictionary matching under SCER by using the automata.
Moreover, we show the time and space complexity of our algorithms with respect to the size of input strings.

%\keywords{Dictionary matching  \and Aho-Corasick algorithm \and Substring consistent equivalence relation.}
\end{abstract}
\section{Introduction}
\emph{The pattern matching problem} is one of the most fundamental problems in string processing and extensively studied due to its wide range of applications~\cite{Crochemore1994,Crochemore2002}.
Given a text $T$ of length $n$ and a pattern $P$ of length $m$, the pattern matching problem is to find all occurrence positions of $P$ in $T$.
A naive approach to solve this problem is by comparing all substrings of $T$ whose length is $m$ to $P$ which takes $O(nm)$ time.
One of the algorithms that can solve this problem in linear time and space is the \emph{Knuth-Morris-Pratt (KMP) algorithm}~\cite{Knuth1977}.
The KMP algorithm constructs an $O(m)$ space array as a failure function by preprocessing the pattern in $O(m)$ time, then uses the failure function to perform pattern matching in $O(n)$ time.

Many variants of pattern matching problems are studied for various applications such as parameterized pattern matching~\cite{Baker1996} for detecting duplication in source code, order-preserving pattern matching~\cite{Kim2014,Kubica2013} for numerical analysis,
permuted pattern matching~\cite{Katsura2013} for multi sensor data,
and so on~\cite{Park2019}.
In order to solve these problems efficiently,
the KMP algorithm is extended for the above mentioned pattern matching problems~\cite{Amir1994,Diptarama2016,Hendrian2019,Kim2014,Kubica2013,Park2019}.

Recently, Matsuoka~\etal{}~\cite{Matsuoka2016} defined a general pattern matching problem under a substring consistent equivalence relation.
An equivalence relation $\approx$ for two strings $X \approx Y$ is a \emph{substring consistent equivalence relation} (\emph{SCER})~\cite{Matsuoka2016}
if for two strings $X,Y$, $X \approx Y$ implies $|X| = |Y|$ and $X[i:j] \approx Y[i:j]$ for all $1 \le i \le j \le |X|$. 
The equivalence relations used in parameterized pattern matching, order-preserving pattern matching, and permuted pattern matching are SCERs.
Matsuoka~\etal{} proposed a generalized KMP algorithm that can solve any pattern matching under SCER
and showed the time complexity of the algorithm.
They also show periodicity properties on strings under SCERs.

The \emph{dictionary matching problem} is a task to find all occurrence positions of multiple patterns in a text.
Given a set of patterns called a dictionary $\DD$,
we can find the occurrence positions of all patterns in a text $T$ by performing pattern matching for each pattern in the dictionary.
However, we need to read the text multiple times in this approach.
Aho and Corasick~\cite{Aho1975} proposed an algorithm that can perform dictionary matching in linear time by extending the failure function of the KMP algorithm.
The Aho-Corasick (AC) algorithm constructs an automaton (we call this automaton as an \emph{AC-automaton}) from $\DD$
and then uses this automaton to find the occurrences of all patterns in the text.
The AC-automaton of $\DD$ uses $O(m)$ space and can be constructed in $O(m\log |\Sigma|)$ time,
where $m$ is the sum of the length of all patterns in $\DD$ and $|\Sigma|$ is the alphabet size.
By using an AC-automaton, all occurrences of patterns in $T$ can be found only by reading $T$ once, which takes $O(n\log |\Sigma|)$ time.
Similarly to the KMP algorithm,
the AC algorithm is also extended 
to perform dictionary matching on some variant of strings~\cite{Diptarama2016fast,Hendrian2019,Idury1996,Katsura2013,Kim2014}.
In order to perform dictionary matching efficiently,
the extended AC algorithms encode the patterns in a dictionary
and create an automaton from the encoded patterns instead of the patterns itself.

\begin{table}[t]
	\caption{The time complexity of the proposed algorithm on some dictionary matching problems.} 
	\label{table:example}
	\centering
	\setlength{\tabcolsep}{3pt}
	\begin{tabular}{|l|c|c|c|c|c|}
		\hline
		& $\EnTime{n}$ & $\REnTime{\ell}$ & $\pi$ & Preprocessing & Searching \\
		\hline
		Exact  & $O(1)$ & $O(1)$ & $|\Sigma|$ & $O(m \log|\Sigma|)$ & $O(n \log|\Sigma| + \OCC)$ \\
		Parameterized  & $O(n)$ & $O(1)$ & $|\Sigma|$ & $O(m \log|\Sigma|)$ & $O(n \log|\Sigma| + \OCC)$  \\
		Order-preserving & $O(1)$ & $O(\log\ell)$ & $\ell$ & $O(m \log \ell)$ & $O(n \log\ell + \OCC)$ \\
	
		\hline
	\end{tabular}
\end{table}

In this paper, we propose a generalization of the Aho-Corasick algorithm for dictionary matching under SCER.
The proposed algorithm encodes the patterns in the dictionary,
and then constructs an automaton with a failure function called a \emph{substring consistent equivalence relation automaton} (\emph{SCER automaton}) from the encoded strings.
We present an algorithm to construct SCER automata and show how to perform dictionary matching by using SCER automata.
Suppose we can encode a string $X$ in $\EnTime{|X|}$ time and
re-encode $X[i:j][j-i+1]$ in $\REnTime{|X[i:j]|}$ time,
we show that the size of SCER automaton is $O(m)$
and can be constructed in $O(\EnTime{m} + m \cdot (\REnTime{\ell}+\log\pi))$ time,
where $m$ is the sum of the length of all patterns in the dictionary, $\ell$ is the length of the longest patterns in the dictionary, and $\pi$ is the maximum number of possible outgoing transitions from any state.
Moreover, we show that the dictionary matching under SCER can be performed in $O(\EnTime{n} + n \cdot (\REnTime{\ell}+\log\pi))$ time by using SCER automata, where $n$ is the length of the text.
By using our algorithm,
we can perform dictionary matching under any SCER.
\Cref{table:example} shows the time complexity of our algorithm on some dictionary matching problems.
\section{Preliminaries}
Let $\Sigma$ and $\Pi$ be integer alphabets, and $\Sigma^*$ (resp. $\Pi^*$) be the set of all strings over $\Sigma$ (resp. $\Pi$).
The empty string $\varepsilon$ is the string of length 0.
We assume that the size of any symbol in $\Sigma$ and $\Pi$ is constant
and a comparison of any two symbols in $\Sigma$ or $\Pi$ can be done in constant time in word RAM model.
For a string $T \in \Sigma^*$,
$|T|$ denotes the length of $T$.
For $1 \le i \le j \le |T|$, $T[i]$ denotes the $i$-th character of $T$
and $T[i:j]$ denotes the \emph{substring} of $T$ that starts at $i$ and ends at $j$.
Let $T[:j] = T[1:j]$ denote the \emph{prefix} of $T$ that ends at $j$ and
$T[i:] = T[i:|T|]$ denote the \emph{suffix} of $T$ that starts at $i$.
For convenience we define $T[i:j] = \varepsilon$ if $j < i$.
Note that $\varepsilon$ is a substring, a prefix, and a suffix of any string.
For a string $T$, let $\Pref(T)$ denote the set of all prefixes of $P$.
For two strings $X$ and $Y$, we denote by $XY$ or $X \cdot Y$ the concatenation of $X$ and $Y$.

Let $\DD = \{P_1,P_2,...,P_d\}$ be a set of patterns, called a \emph{dictionary}.
Let $|\DD|$ denote the number of patterns in $\DD$ and $\|\DD\| = \Sigma_{k=1}^{|\DD|}|P_k|$ denote the total length of the patterns in $\DD$.
For a dictionary $\DD$,
$\Pref(\DD) = \bigcup_{k=1}^{|\DD|}\Pref(P_k)$ is the set of all prefixes of the patterns in $\DD$.

\begin{figure}[t]
	\centering
	\includegraphics[scale=0.6]{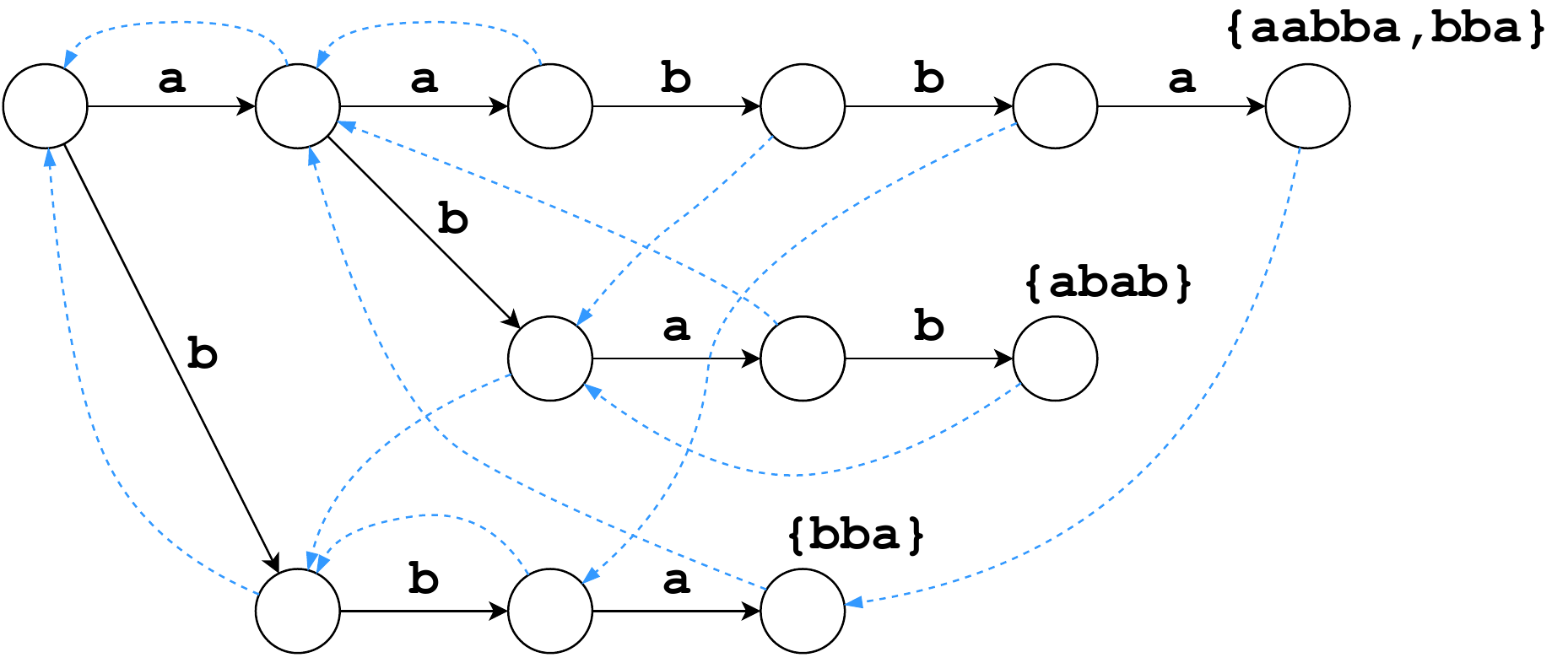}
	\caption{The AC-automaton of $\DD=\{\mathtt{aabba},\mathtt{abab},\mathtt{bba}\}$.
	The solid arrows represent the goto function, the dashed blue arrows represent the failure function, and the sets of strings represent the output function.}
	\label{fig:ac}
\end{figure}

The Aho-Corasick automaton~\cite{Aho1975} $\AC(\DD)$ of $\DD$ consists of a set of states 
and three functions: goto, failure, and output functions.
Each state of $\AC(\DD)$ corresponds to a prefix in $\Pref(\DD)$.
The goto function $\Goto_{\DD}$ of $\AC(\DD)$ is defined so that $\Goto_{\DD}(P_k[:i],P_k[i+1]) = P_k[:i+1]$,
for any $P_k \in \DD$ and $1 \le i < |P_k|$.
The failure fuction $\Fail_{\DD}$ of $\AC(\DD)$ is defined so that $\Fail_{\DD}(P_k[:i]) = P_k[j:i]$ where $j = \min\{l \mid l > 1, P_k[l:i] \in \Pref(\DD)\}$.
The output fuction $\Out_{\DD}$ of $\AC(\DD)$ is defined as $\Out_{\DD}(P_k[:j]) = \{P \in \DD \mid P = P_k[i:j] \mbox{ for some } 1 \le i \le j\}$.
\Cref{fig:ac} shows an example of AC-automaton.
We will define a generalization of AC-automata later in \Cref{sec:scer_automata}.

Next, we define the class of equivalence relations that we consider in this paper called \emph{substring consistent equivalence relations}.
\begin{definition}[Substring consistent equivalence relation (SCER) $\approx$~\cite{Matsuoka2016}]
	An equivalence relation $\approx \ \subseteq \Sigma^* \times  \Sigma^*$ is a \emph{substring consistent equivalence relation (SCER)} if 
	for two string $X$ and $Y$, $X \approx Y$ implies $|X| = |Y|$ and $X[i:j] \approx Y[i:j]$ for all $1 \le i \le j \le |X|$.
\end{definition}
We say $X$ \AMatch es $Y$ iff $X \approx Y$.
For instance, the matching relations in
parameterized pattern matching~\cite{Baker1996}, order-preserving pattern matching~\cite{Kim2014,Kubica2013},
and permuted pattern matching~\cite{Katsura2013} are SCERs,
while the matching relations in indeterminate string pattern matching~\cite{Antoniou2008} and function matching~\cite{Amir2006} are not.

Matsuoka~\etal~\cite{Matsuoka2016} define occurrences of a pattern in a text under an SCER $\approx$, which is used to define the pattern matching under SCERs as follows.
\begin{definition}[\AOccur{}]
	For two strings $T$ and $P$, a position $i$,  $1 \le i \le |T| - |P| +1$, is an \AOccur{} of $P$ in $T$ iff $P \approx T[i:i+|P|-1]$.
\end{definition}
By using the above definition we define the \emph{dictionary matching under SCERs}.
\begin{definition}[\ADic{}]
Given a dictionary $\DD = \{P_1,P_2,\dots,P_d\}$ and a text $T$,
the dictionary matching with respect to an SCER $\approx$ (\ADic)
is a task to find all \AOccurs{} of $P_k$ in $T$ for all $P_k \in \DD$.
\end{definition}

In order to perform some variants of dictionary matching fast,
encodings are used on strings.
For instance,
the prev-encoding is used for parameterized pattern matching~\cite{Idury1996}
and the nearest neighbor encoding is used for order-preserving pattern matching~\cite{Kim2014}.
Following the previous research,
we generalize these encodings for SCERs as follows.
\begin{definition}[\AEn]
Let $\Sigma$ and $\Pi$ be alphabets.
We say an encoding function $f:\Sigma^* \rightarrow \Pi^*$ is a prefix encoding with respect to an SCER $\approx$ (\AEn) if
(1) for a string $X$, $|X| = |f(X)|$,
(2) $f(X[:i]) = f(X)[:i]$, and
(3) for two strings $X$ and $Y$,
	$f(X) = f(Y)$ iff $X \approx Y$.
\end{definition}
We can easily confirm that both the prev-encoding~\cite{Baker1996} and the nearest neighbor encoding~\cite{Kim2014} are prefix encodings.
Amir and Kondratovsky~\cite{Amir2019} show that there exists a prefix encoding for any SCER.
By using a \AEn{},
if $X[:i] \approx Y[:i]$ we can check whether $X[:i+1] \approx Y[:i+1]$
just by checking whether $f(X)[i+1] = f(Y)[i+1]$.
Therefore, \ADic{} can be performed fast by using prefix encoded strings.

For a string $P$ and prefix encoding $f$,
let we denote $f(P)$ by $\En{P}$ for simplicity.
For a dictionary $\DD = \{P_1,P_2,\dots,P_d\}$, 
let $\En{\DD}=\{\En{P_1},\En{P_2},\dots,\En{P_d}\}$
and $\En{\Pref(\DD)} = \Pref(\En{\DD}) = \bigcup_{k=1}^d \Pref(\En{P_k})$.

Throughout the paper, let $T$ be a text of length $n$, $\DD$ be a dictionary, $d = |\DD|$, $m = \|\DD\|$, and $\ell = \max\{|P_k|\mid P_k \in \DD\}$.
Let $\Pi^*$ be the co-domain of a \AEn{}.
For a string $X$, suppose that $\En{X}$ can be computed in $\EnTime{|X|}$ time.
Assuming that $\En{X}$ has been computed,
suppose we can re-encode $\En{X[i:j]}[j-i+1]$ in $\REnTime{|X[i:j]|}$ time.
\section{SCER automata}\label{sec:scer_automata}

In this section, we propose automata for the \ADic{} problem called 
\emph{substring consistent equivalence relation automata (SCERAs)}.
First, we describe the definition and properties of the SCERA for a dictionary $\DD$,
then show the size of the SCERA of $\DD$ with respect to the size of $\DD$.
After that, we propose a \ADic{} algorithm by using SCERAs and
show the time complexity of the proposed algorithm.
Last, we present an algorithm to construct SCERAs and show its time complexity.

\subsection{Definition and properties}\label{subsec:defprop}

For a dictionary $\DD=\{P_1,P_2,\dots,P_d\}$,
the substring consistent equivalence relation automaton $\SCERA(\DD)$ of $\DD$ 
consists of a set of states and three functions, namely \emph{goto}, \emph{failure}, and \emph{output} functions.

The set of states $\mathcal{S}_{\DD}$ of $\SCERA(\DD)$ defined as follows.
\begin{align*}
	\mathcal{S}_{\DD} =\{S \mid S \in \Pref(\En{\DD}) \}
\end{align*}
Each state of $\SCERA(\DD)$ corresponds to a prefix of $\En{P_k}$ for some $P_k \in \DD$,
thus we can identify each state by the corresponded prefix.
Since the number of prefixes of $P_k$ is $|P_k| + 1$ and $|P_k|=|\En{P_k}|$, the number of states of $\SCERA(\DD)$ is as follows.
\begin{lemma}\label{lem:state_num}
	For a dictionary $\DD$, the number of states $|\mathcal{S}_{\DD}|$ of $\SCERA(\DD)$ 
	is $O(m)$.
\end{lemma}

Next, we define the functions in $\SCERA(\DD)$.
First, the goto function $\Goto_{\DD}$ of $\SCERA(\DD)$ is defined as follows.
\begin{definition}[Goto function]
	The goto function $\Goto_{\DD} : \CS_{\DD} \cdot \Pi \rightarrow \CS_{\DD} \cup \{\Null\}$ of $\SCERA(\DD)$ is defined by
	\begin{align*}
		\Goto_{\DD}(S,c) = \begin{cases}
		S\cdot c & \mbox{\textup{if} } (S\cdot c) \in \CS_{\DD}, \\
		\Null & \mbox{\textup{otherwise.}}
		\end{cases}
	\end{align*}
\end{definition}
Intuitively, $\Goto_{\DD}(\En{P_k}[:j],\En{P_k}[j+1]) = \En{P_k}[:j+1]$, for any $P_k \in \DD$ and $0 \le j < |P_k|$.
The states and goto function form a trie of all encoded patterns in $\En{\DD}$.
For two states $S$ and $S'$ such that $S' = \Goto_{\DD}(S,c)$ for some $c \in \Pi$,
we call $S$ \emph{the parent} of $S'$ and $S'$ \emph{a child} of $S$.
For convenience, for a state $S$ and a string $X \in \Pi^*$, let $\Goto_{\DD}(S,X) = \Goto_{\DD}(\Goto_{\DD}(S,X[1]),X[2:])$ and $\Goto_{\DD}(S,\varepsilon) = S$.
Here we denote by $\pi$, the maximum number of possible outgoing transitions from any state.

Next, the failure function $\Fail_{\DD}$ of $\SCERA(\DD)$ is defined as follows.
\begin{definition}[Failure function]
	The failure function $\Fail_{\DD}: \CS_{\DD} \rightarrow \CS_{\DD} \cup \{\Null\}$ of $\SCERA(\DD)$ is defined by $\Fail_{\DD}(\varepsilon) = \Null$ and $\Fail_{\DD}(\En{P_k}[:j]) = \En{P_k[i:j]}$ for $1 \le j \le |P_k|$, where $i = \min \{l \mid l > 1, \En{P_k[l:j]} \in \Pref(\En{\DD})\}$.
\end{definition}
In other words, 
for any state $\En{P_k}[:j]$,
$\Fail_{\DD}(\En{P_k}[:j]) = \En{P_k[i:j]}$
iff $P_k[i:j]$ is the longest suffix of $P_k[:j]$ such that 
 $\En{P_k[i:j]} \in \Pref(\En{\DD})$.
Moreover, by the definition of prefix encoding,
$P_k[i:j]$ is also the longest suffix of $P_k[:j]$
that \AMatch es a prefix in $\Pref(\DD)$.
For convenience, 
assume that $\Fail_{\DD}(\Null) = \Null$, 
we define $\Fail_{\DD}$ recursively,
namely $\Fail_{\DD}^k(S) = \Fail_{\DD}^{k-1}(\Fail(S))$ and $\Fail_{\DD}^0(S) = S$.

The failure function has the following properties.
\begin{lemma}\label{lem:fail_rec}
	For any state $\En{P_k}[:j]$,
	if $\En{P_k[i:j]}$ is a state,
	there is $q \ge 0$ such that $\Fail_{\DD}^q({\En{P_k}}[:j]) = \En{P_k[i:j]}$.
\end{lemma}
\begin{proof}
	Straightforward by induction on $j$.
	\qed
\end{proof}
\begin{lemma}\label{lem:fail_next}
	Consider two states $\En{P_k}[:j]$ and $\En{P_k}[:j+1]$.
	If $\Fail_{\DD}(\En{P_k}[:j+1]) \ne \varepsilon$,
	there is $q \ge 0$ such that $\Fail_{\DD}^q({\En{P_k}}[:j])$ is the parent of $\Fail_{\DD}(\En{P_k}[:j+1])$.
\end{lemma}
\begin{proof}
	Let $\En{P_k[i:j+1]} = \Fail_{\DD}(\En{P_k}[:j+1])$.
	From the condition $\Fail_{\DD}(\En{P_k}[:j+1]) \ne \varepsilon$,
	we have $i \le j+1$.
	Since $\En{P_k[i:j+1]} \in \Pref(\En{\DD})$,
	clearly $\En{P_k[i:j]} \in \Pref(\En{\DD})$
	and $\En{P_k[i:j]}$ is the parent of $\En{P_k[i:j+1]}$.
	By \Cref{lem:fail_rec},
	$\En{P_k[i:j]} = \Fail_{\DD}^q(\En{P_k}[:j])$ for some $q \ge 0$.
	\qed
\end{proof}
\Cref{lem:fail_rec} implies that any re-encoded suffix $\En{P_k[i:j]}$ of $\En{P_k}[:j]$ such that $\En{P_k[i:j]} \in \Pref(\En{\DD})$
can be found by executing failure function recursively starting from $\Fail_{\DD}(\En{P_k}[:j])$.
Moreover, \Cref{lem:fail_next} implies that for any state $\En{P_k}[:j+1]$
such that $\Fail_{\DD}(\En{P_k}[:j+1]) \ne \varepsilon$,
the parent of $\Fail_{\DD}(\En{P_k}[:j+1])$
can be found by executing failure function recursively starting from $\Fail_{\DD}(\En{P_k}[:j])$.

Last, the output function $\Out_{\DD}$ of $\SCERA(\DD)$ is defined as follows.
\begin{definition}[Output function]
	The output function $\Out_{\DD}$ of $\SCERA(D)$ is defined by $\Out_{\DD}(\En{P_k}[:j]) = \{P \in \DD \mid P \approx P_k[i:j] \mbox{ for some } 1 \le i \le j\}$.
\end{definition}
For a state $\En{P_k}[:j]$, $\Out_{\DD}(\En{P_k}[:j])$
is the set patterns that \AMatch{} some suffix of $P_k[:j]$.

The output function has the following properties.
\begin{lemma}\label{lem:out_fail}
	For any $P \in \Out_{\DD}(\En{P_k}[:j])$,
	$\En{P} = \Fail^q_{\DD}(\En{P_k}[:j])$ for some $q \ge 0$.
\end{lemma}
\begin{proof}
	From the definition of $\Out_{\DD}$ and prefix encoding, $P \approx P_k[i:j]$ and $\En{P} = \En{P_k[i:j]}$ for some $i$.
	By \Cref{lem:fail_rec}, 
	$\En{P_k[i:j]} = \Fail_{\DD}^q(\En{P_k}[:j])$ for some $q \ge 0$.
	\qed
\end{proof}
\begin{lemma}\label{lem:out_union}
	For any state $S \ne \varepsilon$,
	if $S \in \En{\DD}$, $\Out_{\DD}(S) = \{S\} \cup \Out_{\DD}(\Fail_{\DD}(S))$.
	Otherwise, $\Out_{\DD}(S) = \Out_{\DD}(\Fail_{\DD}(S))$.
\end{lemma}
\begin{proof}
	Assume there is a pattern $P$ such that $P \in \Out_{\DD}(\Fail_{\DD}(S))$ but  $P \not \in \Out_{\DD}(S)$.
	By \Cref{lem:out_fail}, there exists $q$ such that $\En{P} = \Fail_{\DD}^q(\Fail_{\DD}(S))$.
	Since $\Fail_{\DD}^q((\Fail_{\DD}(S)) = \Fail_{\DD}^{q+1}(S)$,
	we have $P \in \Out_{\DD}(S)$ which contradicts the assumption.
	
	Next, assume there is a pattern $P$ such that $P \in \Out_{\DD}(S)$ but $P \not \in \Out_{\DD}(\Fail_{\DD}(S))$.
	By \Cref{lem:out_fail}, there is $q$ such that $\En{P} = \Fail^q_{\DD}(S)$.
	If $q > 0$, $\En{P} = \Fail^{q-1}_{\DD}(\Fail(S))$ implies $P \in \Out_{\DD}(\Fail_{\DD}(S))$
	which contradicts the assumption.
	Therefore, the remaining possibility is $q = 0$ which implies $S = \En{P} \in \En{\DD}$.
	\qed
\end{proof}
\Cref{lem:out_union} implies that we can compute $\Out_{\DD}(S)$ by copying 
 $\Out_{\DD}(\Fail_{\DD}(S))$ and adding $S$ if $S \in \DD$.
We will utilize \Cref{lem:out_union} to construct the output function efficiently.

\subsection*{Implementation and Space complexity}
We will describe how to implement SCER automata and show its space complexity.
First, The goto function $\Goto_{\DD}$ can be implemented by using an associative array on each state.
We have the following lemma for the required space and time to implement the goto function of $\SCERA(\DD)$.
\begin{lemma}\label{lem:gotosize}
	Assume that the size of any symbol in $\Pi$ is constant.
	The goto function $\Goto_{\DD}$ of $\SCERA(\DD)$ can be implemented in $O(m)$ space.
\end{lemma}
\begin{proof}
	The number of associative arrays used to implement $\Goto_{\DD}$ is $O(|\CS|)$.
	Since for each $S'$ there only exists one pair $(S.c) \in \mathcal{S} \cdot \Pi$ such that $\Goto(S,c) = S'$ ,
	the total size of associative arrays is $O(|\CS|)$.
	Therefore, $\Goto_{\DD}$ can be implemented in $O(m)$ space by \Cref{lem:state_num}.
	\qed
\end{proof}

Next, the failure function can be implemented by using a state pointer on each state.
\begin{lemma}\label{lem:failsize}
	For a dictionary $\DD$,
	$\Fail_{\DD}$ can be implemented in $O(m)$ space.
\end{lemma}
\begin{proof}
	Since $\Fail_{\DD}$ is defined for each state,
	$\Fail_{\DD}$ can be implemented using $O(|\CS|)$ space.
	Therefore, $\Fail_{\DD}$ can be implemented in $O(m)$ space by \Cref{lem:state_num}.
	\qed
\end{proof}

Last, similarly to the original AC-automata,
the output function $\Out_{\DD}$ can be implemented in linear space by using a list.
\begin{lemma}\label{lem:outputsize}
	For a dictionary $\DD$,
	$\Out_{\DD}$ can be implemented in $O(m)$ space.
\end{lemma}
\begin{proof}
	Each state $S$ stores a pair $(i,p)$,
	where $i$ is the pattern number if $S \in \En{\DD}$ or $\Null$ otherwise
	and $p$ is a pointer to a state 
	$S' = \Fail^q_{\DD}(\En{P_k}[:j])$ for the smallest $q>0$
	such that $S' \in \En{\DD}$ if it exists or $\Null$ otherwise.
	Since $|\CS|$ is $O(m)$ by \Cref{lem:state_num},
	$\Out_{\DD}$ can be implemented in $O(m)$ space.
	\qed
\end{proof}
From Lemmas~\ref{lem:state_num}, \ref{lem:gotosize}, \ref{lem:failsize},
and \ref{lem:outputsize}, we get the following theorem.
\begin{theorem}
	Assume that the size of any symbol in $\Pi$ is constant.
	For a dictionary $\DD=\{P_1,P_2,\dots,P_d\}$ of total size $\| \DD \|=m$,
	$\SCERA(\DD)$ can be implemented in $O(m)$ space.
\end{theorem}

\subsection{Dictionary matching using SCERA}

\begin{algorithm2e}[t]
	\caption{A \ADic{} algorithm using $\SCERA(\DD)$}
	\label{alg:match}
	\SetVlineSkip{0.5mm}
	\Fn{$\Match(T)$}{
		compute $\En{T}$\;
		$v \leftarrow \Root$\;
		\For{$i \leftarrow 1$ \bf{to} $|T|$}{
			\lWhile{$\Goto(v, \En{T[i-\Depth(v):i]}[\Depth(v)+1]) = \Null $}{
				$v \leftarrow \Fail(v)$}
			$v \leftarrow \Goto(v, \En{T[i-\Depth(v):i]}[\Depth(v)+1])$\;
			\lFor{$j \in \Out(v)$}{
				\bf{output} $(j,i-|P_j|+1)$}
		}
	}
\end{algorithm2e}

In this section, we describe how to use SCER automata for dictionary matching.
The proposed algorithm for \ADic{} is shown in \Cref{alg:match}. 
For any state $S$, let $\Depth(S)$ be the depth of $S$ i.e. the length of the shortest path from $\Root$ to $S$.
In order to simplify the algorithm we use an auxiliary state $\bot$
where $\Goto(\bot,c) = \Root$ for any $c \in \Pi$, $\Fail(\Root) = \bot$,
and $\Depth(\bot) = 0$.

The algorithm starts with $\Root$ as the active state and $1$ as the active position.
The algorithm reads the encoded text from the left to the right,
while updating the active state and the active position.
Let $v = \En{T[i-\Depth(v):i-1]}$ be the active state and $i$ be the active position.
The algorithm finds $\En{T[i-\Depth(v):i]}[\Depth(v)]$ transition from $v$.
If $\En{T[i-\Depth(v):i]}[\Depth(v)]$ transition exists,
the algorithm updates the active state to $\Goto(v, \En{T[i-\Depth(v):i]}[\Depth(v)])$
and increments $i$.
After the algorithm updates the active state,
it outputs all patterns in $\Out(v)$.
Otherwise if $\En{T[i-\Depth(v):i]}[\Depth(v)]$ transition does not exist, the algorithm updates the active state to $\Fail(v)$
without updating $i$.
The algorithm repeats these operations until it reads all the text.

\begin{lemma}\label{lem:longest_suffix}
	Let $v = \En{T[i-\Depth(v):i-1]}$ be the active state and $i$ be the active position.
	$T[i-\Depth(v):i-1]$ is the longest suffix of $T[:i-1]$
	such that $v \in \Pref(\En{\DD})$.
\end{lemma}

\begin{proof}
	We will prove by induction.
	Initially, $v = \Root$ and $i = 1$,
	thus $v=\varepsilon$ is the longest suffix of $\En{T[:0]} = \varepsilon$.
%%	Next, if $\Goto(v,\En{T[:1]}[1]) \ne \Null$,
%	$v$ is updated to $\En{T[1:1]}$.
%	Clearly $T[1:1]$ is the longest suffix of $T[:1]$ and $\En{T[1:1]} \in \Pref(\En{\DD})$.
%	Otherwise if $\Goto(v,\En{T[:1]}[1]) = \Null$,
%	$v$ is updated to $\Root$.
%	Since $\Goto(v,\En{T[:1]}[1]) = \Null$ implies $\En{T[1:1]} \not \in \Pref(\En{\DD})$, $\varepsilon$ is the longest suffix of $T[:1]$ that is also an element of $\Pref(D)$.	
	Assume that $v = T[i-\Depth(v):i-1]$ is the longest suffix of $T[:i-1]$
	such that $v \in \Pref(\En{\DD})$.
	Let $u$ be the next active state and $i$ be the next active position.
%	If $\Goto(v,\En{T[i-\Depth(v):i]}[\Depth(v)]) \ne \Null$,
%	we have $u = \Goto(v,\En{T[i-\Depth(v):i]}[\Depth(v)]) = \En{T[i-\Depth(u):i]}$.
%	By the assumption, $T[i-\Depth(u):i]$ is the longest suffix of $T[:i]$
%	such that $u \in \Pref(\En{\DD})$.
%	Otherwise if $\Goto(v,\En{T[i-\Depth(v):i]}[\Depth(v)]) = \Null$,
	From the algorithm, $u = \Goto(\Fail^q(v),\En{T[i-\Depth(\Fail^q(v)):i]}[\Depth(\Fail^q(v))+1])$,
	where $q$ is the smallest integer such that
	$\Goto(\Fail^q(v),\En{T[i-\Depth(\Fail^q(v)):i]}[\Depth(\Fail^q(v))+1]) \ne \Null$.
	Let $T[i-j:i]$ be the longest suffix of $T[:i]$ such that $\En{T[i-j:i]} \in \Pref(\En{\DD})$.
	If $\En{T[i-j:i]} = \varepsilon$,
	we have $\Goto(\En{T[i-l:i-1]},\En{T[i-l:i]}[l+1]) = \Null$ for any $l$, $0 \le l \le \Depth(v)$, such that $\En{T[i-l:i-1]} \in \Pref(\En{\DD})$.
	Thus, we have $\Fail^q(v) = \bot$ and $u = \Root$.
	Otherwise if $\En{T[i-j:i]} \ne \varepsilon$,
	we have $\En{T[i-j:i-1]} \in \Pref(\En{\DD})$.
	Moreover, $\En{T[i-j:i-1]} = \Fail^{q'}(v)$ for some ${q'}$ by \Cref{lem:fail_rec} and $\Goto(\En{T[i-j:i-1]},\En{T[i-j:i]}[j+1]) \ne \Null$.
	Since $T[i-j:i]$ is the longest suffix of $T[:i]$ such that $\En{T[i-j:i]} \in \Pref(\En{\DD})$, $\Goto(\En{T[i-l:i-1]},\En{T[i-l:i]}[l+1]) = \Null$, for any $l$, $j < l \le \Depth(v)$, such that $\En{T[i-l:i-1]} \in \Pref(\En{\DD})$.
	Therefore, $q'=q$ which implies the correctness of \Cref{lem:longest_suffix}.
	\qed
\end{proof}

\begin{theorem}
	Given $\SCERA(\DD)$ and a text $T$ of length $n$, 
	\Cref{alg:match} outputs all occurrence positions of all patterns in $T$ correctly
	in $O(\EnTime{n} + n \cdot (\REnTime{\ell}+\log\pi) + \OCC)$ time,
	where $\ell$ is the length of the longest pattern in $\DD$,
	$\EnTime{n}$ is the time required to encode $T$,
	$\REnTime{\ell}$ is the time required to re-encode a symbol of a substring of $T$ whose length is $\ell$ or less,
	$\pi$ is the maximum number of possible outgoing transitions from any state,
	and $\OCC$ is the number of occurrences of the patterns in $T$.
\end{theorem}
\begin{proof}
	First, we show the correctness of the algorithm.
	Assume there is an occurrence position $o$ of $P_k$ in $T$ that has not been output by the algorithm.
	Let $o+|P_k|$ be the active position.
	By \Cref{lem:longest_suffix}, 
	the active state $v = \En{T[j:o+|P_k|-1]}$ is the longest suffix of $T[:o+|P_k|-1]$
	such that $\En{T[j:o+|P_k|-1]} \in \Pref(\En{D})$.
	Since $\En{T[o:o+|P_k|-1]} \in \Pref(\En{D})$, we have $j \le o$.
	By the definition of the output function,
	$P_k \in \Out(\En{T[j:o+|P_k|-1]})$.
	Therefore, the algorithm outputs $P_k \in \Out(\En{T[j:o+|P_k|-1]})$ which contradicts the assumption.
	Moreover, for any active state $\En{T[i:j]}$ the algorithm only outputs an occurrence position of $P_k$ iff $P_k$ \AMatch es a suffix of $T[i:j]$
	by \Cref{lem:out_union}.
	
	Next, we show the time complexity of the algorithm.
	The encoding $\En{T}$ can be computed in $\EnTime{n}$.
	For each position $i$, the depth of the active position is increased by one,
	thus the depth of the active position is increased by $n$ in total.
	Since the depth of the active position is decreased by at least one each time $\Fail$ is executed,
	$\Fail$ is executed at most $n$ times.
	Next, each time $\Goto$ is executed, either the depth of the active position is increased by one or $\Fail$ is executed,
	thus $\Goto$ is executed $O(n)$ times.
	Since we need to re-encode a symbol each time $\Goto$ is executed and 
	$\Goto$ can be executed in $O(\log\pi)$ time by binary search,
	the algorithm takes $O(n \cdot (\REnTime{\ell}+\log\pi))$ time to execute $\Goto$ in total.
	In order to output the occurrence positions, 
	the algorithm takes $O(n)$ time to check whether there is any occurrence
	and $O(\OCC)$ time to output the occurrence positions.
	\qed
\end{proof}
\subsection{Constructing SCERA}

\begin{algorithm2e}[t]
	\caption{Computing goto function of $\SCERA(\DD)$}
	\label{alg:goto}
	\SetVlineSkip{0.5mm}
	\Fn{$\ConstGoto(\DD)$}{
		compute $\En{\DD}$\;
		create states $\Root$ and $\bot$\;
		$\Depth(\Root) \leftarrow 0$;
		$\Depth(\bot) \leftarrow 0$\;
		$\Goto(\bot,c) \leftarrow \Root$ for any symbol $c$\;
		\For{$k \leftarrow 1$ \bf{to} $d$}{
			$v \leftarrow \Root$\;
			\For{$j \leftarrow 1$ \bf{to} $|P_k|$}{
				\lIf{$\Goto(v,\En{P_k}[j]) \ne \Null$}{
					$v \leftarrow \Goto(v,\En{P_k}[j])$}
				\Else{
					create a state $u$\;
					$\Goto(v,\En{P_k}[j]) \leftarrow u$\;
					$\Label(u) \leftarrow i$;
					$\Depth(u) \leftarrow \Depth(v) + 1$\;
					$v \leftarrow u$\;
				}
			}
		}
	}
\end{algorithm2e}

\begin{algorithm2e}[t]
	\caption{Computing failure function of $\SCERA(\DD)$}
	\label{alg:failure}
	\SetVlineSkip{0.5mm}
	\Fn{$\ConstFail(\DD)$}{
		compute $\En{\DD}$;
		$\Fail(\Root) \leftarrow \bot$;
		push $\Root$ to $\Queue$\;
		\While{$\Queue \ne \emptyset$}{
			pop $v$ from $\Queue$\;
			\For{$c$ \textup{such that} $\Goto(v,c) \ne \Null$}{
				$u \leftarrow \Goto(v,c)$;
				push $u$ to $\Queue$\;
				$s \leftarrow \Fail(v)$;
				$k \leftarrow \Label(u)$\;
				\While{$\Goto(s, \En{P_k[\Depth(u)-\Depth(s):\Depth(u)]}[\Depth(s)+1]) = \Null $}{
					$s \leftarrow \Fail(s)$\;
				}
				$\Fail(u) \leftarrow \Goto(s,\En{P_k[\Depth(u)-\Depth(s):\Depth(u)]}[\Depth(s)+1)]))$\;
			}
		}
	}
\end{algorithm2e}

\begin{algorithm2e}[t]
	\caption{Computing output function of $\SCERA(\DD)$}
	\label{alg:output}
	\SetVlineSkip{0.5mm}
	\Fn{$\ConstOut(\DD)$}{
		compute $\En{\DD}$\;
		\For{$k \leftarrow 1$ \bf{to} $d$}{
			$v \leftarrow \Root$\;
			\For{$j \leftarrow 1$ \bf{to} $|P_k|$}{
				$v \leftarrow \Goto(v,\En{P_k}[j])$\;
				\lIf{$j = |P_k|$}{
					$\Out(v) \leftarrow \Out(v) \cup \{k\}$}
			}
		}
		push $\Root$ to $\Queue$\;
		\While{$\Queue \ne \emptyset$}{
			pop $v$ from $\Queue$;
			$\Out(v) \leftarrow \Out(v) \cup \Out(\Fail(v))$\;
			\For{$c$ \textup{such that} $\Goto(v,c) \ne \Null$}{
				$u \leftarrow \Goto(v,c)$;
				push $u$ to $\Queue$\;
			}
		}
	}
\end{algorithm2e}

In this section, we describe an algorithm to construct $\SCERA(\DD)$.
We divide the algorithm into three parts: 
goto function, failure function, and output function construction algorithms.

First, the goto function construction algorithm is shown in \Cref{alg:goto}.
Initially, the algorithm computes $\En{P_k}$ for all $P_k \in \DD$;
then it constructs the root state $\Root$ and the auxiliary state $\bot$.
Next, for each pattern $\En{P_k} \in \En{\DD}$,
the algorithm finds the longest prefix of $\En{P_k}$ that exists in the current automaton.
After that, the algorithm creates states corresponding to the remaining prefixes
from the shortest to the longest.
After creating each state, the algorithm updates the goto function,
adds a label $i$ to the state, and compute the depth of the state.

\begin{lemma}\label{lem:gototime}
	Given a dictionary $\DD$, 
	\Cref{alg:goto} constructs the goto function of $\SCERA(\DD)$
	in $O(\EnTime{m} + m\log\pi)$ time.
\end{lemma}
\begin{proof}
	By assumption, each pattern $|P_k|$ in the dictionary can be encoded in $\EnTime{|P_k|}$ time.
	The operations in the inner loop are executed $O(m)$ times.
	$\Goto(v,c)$ can be computed in $O(\log\pi)$ by binary search.
	\qed
\end{proof}

Next, we describe how to compute the failure function of $\SCERA(\DD)$.
\Cref{alg:failure} shows the algorithm for computing the failure function.
The algorithm computes the failure function recursively by the breadth-first search.
The algorithm uses the property in \Cref{lem:fail_next} to compute the failure function.

Consider computing $\Fail(\En{P_k}[:j])$.
Since the algorithm computing $\Fail$ by the breadth-first search,
$\Fail(\En{P_k}[:j-1])$ has been computed.
By \Cref{lem:fail_next}, there is $q \ge 1$ such that $\Fail_{\DD}^q({\En{P_k}}[:j-1])$ is the parent of $\Fail_{\DD}(\En{P_k}[:j])$
or $\Fail_{\DD}(\En{P_k}[:j]) = \varepsilon$.
We can find $q$ by executing the failure function recursively from $\En{P_k}[:j-1]$
and checking whether $\Goto(\En{P_k[i:j]},\En{P_k[i:]}[j-i+1]) = \Null$.

\begin{lemma}\label{lem:failtime}
	Given a dictionary $\DD$ and goto function of $\SCERA(\DD)$, 
	\Cref{alg:failure} constructs the failure function of $\SCERA(\DD)$
	in $O(\EnTime{m} + m \cdot (\REnTime{\ell}+\log\pi))$ time.
\end{lemma}
\begin{proof}
	The dictionary can be encoded in $\EnTime{m}$ time.
	The running time of \Cref{alg:failure} is bounded by the number of executions of $\Fail$.
	Let $x_{k,j}$ be the number of executions of $\Fail$ when finding $\Fail(\En{P_k[:j]})$.
	Since $x_{k,j} \le \Depth(\Fail(\En{P_k[:j-1]})) - \Depth(\Fail(\En{P_k[:j]})) + 2$,
	we can compute $\Sigma_{k=1}^d \Sigma_{j=1}^{|P_k|} x_{k,j} \le \Sigma_{k=1}^d 2|P_k| = 2m$.
	The goto function $\Goto$ is executed each time $\Fail$ be executed.
	Since we need to re-encode a substring each time $\Goto$ be executed and 
	$\Goto$ can be executed in $O(\log\pi)$ time,
	the algorithm takes $O(\EnTime{m} + m \cdot (\REnTime{\ell}+\log\pi))$ time in total.
	\qed
\end{proof}

Last, \Cref{alg:output} shows an algorithm to compute the output function of $\SCERA(\DD)$.
The algorithm first adds $k$ to $\Out(\En{P_k})$ for each $P_k \in \DD$.
Next, the algorithm updates the output function recursively by the breadth-first search.
The algorithm uses the property in \Cref{lem:out_union} to compute the output function.

Consider computing $\Out(\En{P_k}[:j])$.
Since the algorithm computes $\Out$ by the breadth-first search,
$\Out(\Fail(\En{P_k}[:j]))$ has been computed.
By \Cref{lem:out_union},
we can compute $\Out(\En{P_k}[:j])$ by just adding $k$ to $\Out(\Fail(\En{P_k}[:j]))$ if $j = |P_k|$
or by copying $\Out(\Fail(\En{P_k}[:j]))$ if $j \ne |P_k|$.
Note that it can be done efficiently by using pointers as described in Section~\ref{subsec:defprop} instead of copying the set.

\begin{lemma}\label{lem:outtime}
	Given a dictionary $\DD$, the goto function, and the failure function of $\SCERA(\DD)$, 
	\Cref{alg:output} constructs the output function of $\SCERA(\DD)$
	in $O(\EnTime{m} + m \cdot \log\pi)$ time.
\end{lemma}
\begin{proof}
	The dictionary can be encoded in $\EnTime{m}$ time.
	Clearly the loops are executed $O(|\CS|)$ times in total.
	The goto function can be executed in $O(\log\pi)$ time.
	Therefore, \Cref{alg:output} runs in $O(\EnTime{m} + m \cdot \log\pi)$ time.
	\qed
\end{proof}

From Lemmas~\ref{lem:gototime}, \ref{lem:failtime}, and \ref{lem:outtime},
we get the following theorem.
\begin{theorem}
	Given a dictionary $\DD$,
	$\SCERA(\DD)$ can be constructed in $O(\EnTime{m} + m \cdot (\REnTime{\ell}+\log\pi))$ time,
	where $\ell$ is the length of the longest pattern in $\DD$,
	$\EnTime{m}$ is the time required to encode $\DD$,
	$\REnTime{\ell}$ is the time required to re-encode the last symbol of a substring of $P_k \in \DD$,
	and $\pi$ is the maximum number of possible outgoing transitions from any state.
\end{theorem}

\bibliographystyle{splncs04}
\bibliography{ref}

\end{document}